\documentclass[11pt,english]{article}
\usepackage[margin=1in]{geometry}
\usepackage{bbm}
\usepackage{graphicx,hyperref,color}
\usepackage{amsmath,amssymb,amsthm}
\usepackage{enumerate}
\usepackage{enumitem}
\usepackage{fullpage}
\usepackage{algorithm}
\usepackage{mathrsfs}
\usepackage[noend]{algcompatible}
\usepackage[noblocks]{authblk}
\usepackage{varwidth}
\usepackage{mathrsfs}
\usepackage{tcolorbox}
\usepackage{babel}
\usepackage{accents}
\usepackage{wrapfig}
\newtheorem{theorem}{Theorem}[section]

\newtheorem{definition}[theorem]{Definition}
\newtheorem{claim}[theorem]{Claim}
\newtheorem{lemma}[theorem]{Lemma}

\newtheorem{corollary}[theorem]{Corollary}

\newtheorem{observation}[theorem]{Observation}

\def\boldhead#1:{\par\vskip 7pt\noindent{\bf #1:}\hskip 10pt}
\def\ithead#1:{\par\vskip 7pt\noindent{\it #1:}\hskip 10pt}

\def\inline#1:{\par\vskip 7pt\noindent{\bf #1:}\hskip 10pt}
\def\midinline#1:{\par\noindent{\bf #1:}\hskip 10pt}
\def\dnsinline#1:{\par\vskip -7pt\noindent{\bf #1:}\hskip 10pt}
\def\ddnsinline#1:{\newline{\bf #1:}\hskip 10pt}
\def\largeinline#1:{\par\vskip 7pt\noindent{\large\bf #1:}\hskip 10pt}
\long\def\commhide #1\commhideend{}
\long\def\commfull #1\commend{#1}
\long\def\commabs #1\commenda{}
\long\def\commtim #1\commendt{#1}
\long\def\commb #1\commbend{}
\long\def\commedit #1\commeditend{} 

\long\def\commB #1\commBend{}       

\long\def\commex #1\commexend{}     

\long\def\commsiena #1\commsienaend{}  

\long\def\commBI #1\commBIend{}  

\long\def\CProof #1\CQED{}

\def\blackslug{\hbox{\hskip 1pt \vrule width 4pt height 8pt
    depth 1.5pt \hskip 1pt}}
\def\QED{\quad\blackslug\lower 8.5pt\null\par}

\long\def\PPP#1{\noindent{\bf Proof:}{ #1}{\quad\blackslug\lower 8.5pt\null}}

\long\def\denspar #1\densend
{#1}

\newif\ifnotesw\noteswtrue
   {\ifnotesw\marginpar[\hfill\(\top\)]{\(\top\)}\fi}%
   {\ifnotesw\marginpar[\hfill\(\bot\)]{\(\bot\)}\fi}

\newcommand{\mnote}[1]%
    {\ifnotesw\marginpar%
        [{\scriptsize\it\begin{minipage}[t]{\marginparwidth}
        \raggedleft#1%
                        \end{minipage}}]%
        {\scriptsize\it\begin{minipage}[t]{\marginparwidth}
        \raggedright#1%
                        \end{minipage}}%
    \fi}

\def\MathF{\hbox{\rm I\kern-2pt F}}
\def\MathP{\hbox{\rm I\kern-2pt P}}
\def\MathR{\hbox{\rm I\kern-2pt R}}
\def\MathZ{\hbox{\sf Z\kern-4pt Z}}
\def\MathN{\hbox{\rm I\kern-2pt I\kern-3.1pt N}}
\def\MathC{\hbox{\rm \kern0.7pt\raise0.8pt\hbox{\footnotesize I}
\kern-4.2pt C}}
\def\MathQ{\hbox{\rm I\kern-6pt Q}}

\newsavebox{\ttop}\newsavebox{\bbot}

\def\eps{\epsilon}
\def\epsi{\varepsilon}

\newcommand{\mst}{\mathsf{MST}}
\newcommand{\poly}{\mathrm{poly}}
\def\eps{\epsilon}

\newcommand{\child}{\mathsf{child}}
\newcommand{\desc}{\mathsf{desc}}

\newcommand{\stp}{\mathsf{STP}}

\newcommand{\vw}{\mathsf{vw}}
\newcommand{\VD}{\mathsf{VD}}
\newcommand{\aw}{\mathsf{aw}}
\newcommand{\D}{\mathsf{D}}
\newcommand{\AD}{\mathsf{AD}}

\newcommand {\ignore} [1] {}

\title{Light Euclidean Spanners with Steiner Points}
\author{Hung Le}
\affil{University of Victoria and University of Massachusetts Amherst}
\author{Shay Solomon}
\affil{Tel Aviv University}
\date{}
\begin{document}
\maketitle
\thispagestyle{empty}

\maketitle
\begin{abstract}
	The FOCS'19 paper of Le and Solomon \cite{LS19}, culminating a long line of research on Euclidean spanners, 
	proves that the lightness  (normalized weight) of the greedy $(1+\eps)$-spanner in $\mathbb{R}^d$ is 
	$\tilde{O}(\eps^{-d})$ for any $d = O(1)$ and any $\epsilon = \Omega(n^{-\frac{1}{d-1}})$
	(where $\tilde{O}$ hides polylogarithmic factors of $\frac{1}{\epsilon}$),
	and also shows the existence of point sets in $\mathbb{R}^d$ for which any $(1+\eps)$-spanner must have lightness $\Omega(\eps^{-d})$.\footnote{The lightness of a spanner is the ratio of its weight and the MST weight.}
	Given this tight bound on the lightness, a natural arising question is whether a better lightness bound can be achieved using {\em Steiner points}.
	
	Our first result is a construction of Steiner spanners  in $\mathbb{R}^2$ with lightness $O(\eps^{-1} \log \Delta)$, where $\Delta$ is the spread of the point set.\footnote{The spread $\Delta = \Delta(P)$ of a point set $P$ in $\mathbb{R}^d$ is the ratio of the largest to the smallest pairwise distance.}
	In the regime 
	of $\Delta \ll 2^{1/\eps}$, this provides an improvement over the lightness bound of \cite{LS19}; 
	this regime of parameters is of practical interest, as point sets arising in real-life applications 
	(e.g., for various random distributions) have polynomially bounded spread, while in spanner applications $\eps$ often controls the precision, and it sometimes needs to be much smaller than $O(1/\log n)$. 
	Moreover, for spread polynomially bounded in $1/\eps$, this upper bound 
	provides a quadratic improvement over the non-Steiner bound of \cite{LS19}, 
	We then demonstrate that such a light spanner can be constructed in $O_\eps(n)$ time for polynomially bounded spread,
	where $O_\eps$ hides a factor of $\poly(\frac{1}{\epsilon})$.
	Finally, we extend the construction to higher dimensions, proving a lightness upper bound of $\tilde{O}(\epsilon^{-(d+1)/2} + \epsilon^{-2}\log \Delta)$ for any $3\le d = O(1)$ and any $\epsilon = \Omega(n^{-\frac{1}{d-1}})$.
\end{abstract}

\section{Introduction}

A \emph{$t$-spanner} for 
a set  $P$ of points in the $d$-dimensional Euclidean space $\mathbb{R}^d$
is a \emph{geometric graph} that
preserves all the pairwise Euclidean distances between points in $P$ to within a factor of $t$, called the \emph{stretch factor}; by geometric graph we mean a weighted graph in which the vertices correspond to points in $\mathbb{R}^d$ and the edge weights are the Euclidean distances between the corresponding points.
The study of Euclidean spanners dates back to the seminal work of Chew~\cite{Chew86,Chew89} from 1986, which presented a
spanner with constant stretch and $O(n)$ edges for any set of $n$ points in $\mathbb{R}^2$. 
In the three  following decades, Euclidean spanners have evolved into an important subarea of Discrete and Computational Geometry, having found applications in many different areas, such as  approximation algorithms \cite{RS98}, geometric distance oracles~\cite{GLNS02,GLNS02b,GNS05,GLNS08}, and network design \cite{HP00,MP00};
see the book by Narasimhan and Smid ``Geometric Spanner Networks'' \cite{NS07} for an excellent account on 
Euclidean spanners and some of their  applications.
Numerous constructions of Euclidean spanners in two and higher dimensions were introduced over the years, such 
as Yao graphs~\cite{Yao82}, $\Theta$-graphs~\cite{Clarkson87,Keil88,KG92,RS91}, the (path-)greedy spanner~\cite{ADDJS93,CDNS92,NS07} and the gap-greedy spanner~\cite{Salowe92,AS97}--- unveiling an abundance of techniques, tools and insights along the way; refer to the book of~\cite{NS07} for more spanner constructions.

In addition to low stretch, many applications require that the spanner would be {\em sparse}, 
in the unweighted and/or weighted sense.
The {\em sparsity} (respectively, {\em lightness}) of a spanner is the ratio of its size (respectively, weight)
to the size (resp., weight) of a spanning tree (resp., MST),
providing a normalized notion of size (resp., weight), which should ideally be $O(1)$.
For any dimension $d = O(1)$, spanners with constant sparsity are known since the 80s 
\cite{Yao82,Clarkson87,Keil88,KG92,RS91,ADDJS93,Salowe92,AS97};
also, it is known since the early 90s that the greedy spanner has constant lightness \cite{ADDJS93}.
The constant bounds on the sparsity and lightness depend on both $\eps$ and $d$.
In some applications, $\eps$ must be a very small sub-constant parameter, so as to achieve the highest possible precision and minimize potential errors, 
and in some situations $\eps$ may be as small as $n^{-c}$ for some constant $0 < c < 1$. 
Consequently, besides the theoretical appeal,
achieving the precise dependencies on $\eps$ and $d$ in the sparsity and lightness bounds is of practical importance.

Culminating a long line of work, in FOCS'19 Le and Solomon~\cite{LS19} showed that the precise dependencies on $\epsilon$ in the sparsity and lightness bounds 
are $\Theta(\epsilon^{1-d})$ and $\tilde{\Theta}(\epsilon^{-d})$, respectively, for any $d = O(1)$ and any $\epsilon = \Omega(n^{-\frac{1}{d-1}})$; throughout we shall use $\tilde{O},\tilde{\Omega},\tilde{\Theta}$ to hide polylogarithmic factors of $\frac{1}{\epsilon}$.
The lower bounds of \cite{LS19} are proved for the $d$-dimensional sphere, for $d = O(1)$.
On the upper bound side, sparsity $O(\epsilon^{1-d})$ is achieved by a number of classic constructions such as Yao graphs~\cite{Yao82}, $\Theta$-graphs~\cite{Clarkson87,Keil88,KG92,RS91}, the (path-)greedy spanner~\cite{ADDJS93,CDNS92,NS07} and the gap-greedy spanner~\cite{Salowe92,AS97}, and the argument underlying all these upper bounds is basic and simple.
On the other hand, constant lightness upper bound (regardless of the dependency on $\eps$ and $d$) 
is achieved only by the greedy algorithm \cite{ADDJS93,DHN93,DNS95,RS98,NS07,Gottlieb15,BLW19,LS19}, and all known arguments for constant lightness are highly nontrivial, even for $d = 2$; in fact, the proofs in \cite{ADDJS93,DHN93,DNS95,RS98} have missing details.
The first complete proof was given in the book of \cite{NS07}, in a 60-page chapter, where it is shown that the greedy $(1+\eps)$-spanner has lightness $O(\eps^{-2d})$, which improved the dependencies on $\eps$ and $d$ given in all previous work.
In SODA'19, Borradaile, Le and Wulff-Nilsen~\cite{BLW19} presented a shorter and arguably simpler alternative proof that applies to the wider family of {\em doubling metrics}, which improves the $\eps$ dependency provided in FOCS'15 by Gottlieb \cite{Gottlieb15} for doubling metrics, but is inferior to the lightness bound of $O(\eps^{-2d})$ by \cite{NS07}.\footnote{The \emph{doubling dimension} of a metric space $(X,\delta)$ is the smallest value $d$
	such that every ball $B$ in the metric space can be covered by at most
	$2^{d}$ balls of half the radius of $B$.
	This notion generalizes the Euclidean dimension, since the doubling dimension
	of the Euclidean space $\mathbb R^d$ is $\Theta(d)$.
	A metric space is called \emph{doubling} if its doubling dimension is constant.}
Finally, the lightness bound of $\tilde{O}(\epsilon^{-d})$ of~\cite{LS19}
was proved via a tour-de-force argument; interestingly, the proof for $d = 2$ is much more intricate than for higher dimensions $d \ge 3$.

Le and Solomon \cite{LS19} showed that, counter-intuitively, 
one can use {\rm Steiner points} to bypass the sparsity lower bound of Euclidean spanners. 
Specifically, in this way they reduced the sparsity upper bound almost quadratically to $\tilde{O}(\epsilon^{(1-d)/2})$, 
and also provided a matching lower bound for $d = 2$.
Their lower bound for sparsity is derived from a lightness lower bound; specifically, they first prove that, for a set of points evenly spaced on the boundary of the unit square, with distances $\Theta(\sqrt{\eps})$ between neighboring points, any Steiner $(1+\epsilon)$-spanner must incur lightness $\tilde{\Omega}(\frac{1}{\epsilon})$, and then 
they translated the lightness lower bound into a sparsity lower bound of $\tilde{\Omega}(\sqrt{1/\eps})$. 
Whether or not Steiner points can be used to reduce the lightness remained open in \cite{LS19}, 
but this should not come as a surprise.\footnote{The lightness of Steiner spanners can be defined with respect to the SMT (Steiner minimum tree) weight, but we can also stick to the original definition, since the SMT and MST weights differ by a constant factor smaller than 2.} 
First, bounding the lightness of spanners is inherently more difficult than bounding their sparsity--- this is true
for both Euclidean spaces and doubling metrics, as well as other graph families including general weighted graphs
\cite{ADDJS93,DHN93,DNS95,RS98,NS07,Gottlieb15,BLW19,LS19,CW18,FN18,DBLP:journals/siamdm/ElkinNS15}. 
Second, constructing Steiner trees and spanners {\em with asymptotically improved bounds} is inherently more difficult than constructing their non-Steiner counterparts~\cite{DBLP:journals/siamdm/ElkinS11,ES15SIAM,Solomon15,Klein06,BW16,Le20}.  
In our particular case, while the classic sparsity upper bound for non-Steiner Euclidean spanners is simple, its improved counterpart for Steiner spanners by \cite{LS19} requires a number of nontrivial insights and is rather intricate. 
On the other hand, the lightness upper bound of \cite{LS19} uses a tour-de-force argument, and as mentioned already this is true even for $d = 2$.
Consequently, obtaining an improved lightness bound using Steiner points 
seems currently out of reach, at least until an inherently simpler proof to \cite{LS19} for non-Steiner lightness is found (if one exists). 

\subsection{Our Contribution.}
In this paper, we explore the power of Steiner points in reducing lightness for Euclidean spaces of bounded {\rm spread} $\Delta$.\footnote{The {\em spread} of a Euclidean space is the ratio of the maximum to minimum pairwise distances in it.}
Point sets of bounded spread have been studied extensively for Euclidean spanners and related geometric objects
\cite{BHM59,Karp77,CM15,HS05,BM09,Buchin08,Erickson04,NR19,AH10,AFPVX17,ABCGHV05,HarPeled11,HR15,CFLSS18,NR15,NR17,BCIS05,ES10,CLNS15}.
The motivation for studying point sets of bounded spread is three-fold.
\begin{enumerate}
	\item
	Such point sets arise naturally in practice, and are thus important in their own right; indeed, for many random distributions, the spread is polynomial in the number of points--- in expectation and with high probability.
	In particular, it is known that for $n$-point sets drawn uniformly at random from the unit square,
	the expected spread is $\Theta(n)$, and the expected spread in the unit $d$-dimensional hypercube 
	is $n^{2/d}$ for any $d = O(1)$. 
	Researchers have studied random distributions of point sets, in part to explain the success of solving various geometric optimization problems in practice~\cite{BHM59,Karp77,CM15}, and there are many results on spanners for random point sets 
	\cite{DBLP:journals/ipl/Chandra94,DBLP:conf/focs/AryaMS94,DBLP:conf/esa/FarshiG05,DBLP:journals/comgeo/BenkertWWS06,DBLP:conf/wea/FarshiG07,
		DBLP:conf/wads/BoseCCSX07,NS07,ES10,DBLP:journals/comgeo/BoseDLSV11,CLNS15,DBLP:conf/wads/Bar-OnC17,DBLP:journals/corr/abs-1901-08805}.
	Of course, the family of bounded spread point sets is much wider than that of random point sets.
	\item Euclidean spanners can be constructed in (deterministic) $O(n)$ time in such point sets~\cite{Chan08}, while there is no $o(n \log n)$-time algorithm for constructing Euclidean spanners in arbitrary point sets.
	(The result of \cite{Chan08} extends to other geometric objects, such as WSPD and compressed quadtrees, 
	and some of the aforementioned references (e.g., \cite{HR15,BM09}) build on this result to achieve faster algorithms for  other geometric problems for point sets of polynomially bounded spread.) 
	\item The case of bounded spread is sometimes used as a stepping stone towards the general case;
	see, e.g., \cite{GGN06,DBLP:conf/stacs/GudmundssonNS05,DBLP:journals/talg/GudmundssonLNS08,CLNS15,CGMZ16,DBLP:conf/esa/AlstrupDFSW19,LS19}.
\end{enumerate}

As mentioned above, Le and Solomon~\cite{LS19} showed that for a set of points evenly spaced on the boundary of the unit square, with distances $\Theta(\sqrt{\eps})$ between neighboring points, any Steiner $(1+\epsilon)$-spanner must incur lightness $\tilde{\Omega}(\frac{1}{\epsilon})$. 
Note that the spread of this point set is $\frac{1}{\sqrt{\epsilon}}$. 
The non-Steiner upper bound for $d = 2$ by \cite{LS19} is $\tilde{O}(\eps^{-2})$.
A natural arising question is whether one can improve the lightness upper bound of $\tilde{O}(\eps^{-2})$ using Steiner points, ideally quaratically, for point sets of bounded spread. We answer this question in the affirmative.

\begin{theorem}\label{thm:bounded-spread} Any point set $P$ in $\mathbb{R}^2$ of spread $\Delta$ 
	admits a Steiner $(1+\epsilon)$-spanner of lightness
	$O(\frac{\log \Delta}{\epsilon})$.
\end{theorem}

Recalling that the lower bound of $\tilde{\Omega}(\frac{1}{\epsilon})$ by 
\cite{LS19} applies to a point set of spread $\frac{1}{\sqrt{\epsilon}}$,
the lightness upper bound provided by Theorem \ref{thm:bounded-spread} is therefore tight (up to polylogarithmic factors in $\frac{1}{\eps}$)
for point sets of spread $\mathtt{\poly}(\frac{1}{\eps})$.
Moreover, this upper bound improves the general upper bound of  $\tilde{O}(\eps^{-2})$ from \cite{LS19}
in the regime $\log \Delta \ll \eps^{-1}$, i.e., when $\Delta \ll 2^{1/\eps}$. 
For spread polynomial in $n$, we get an improvement over \cite{LS19} as long as $\eps \ll \frac{1}{\log n}$.
Of course, the improvement gets more significant as $\eps$ decays--- in the most extreme situation $\eps$ is inverse polynomial in $n$, and then the improvement over \cite{LS19} is polynomial in $n$ even when $\Delta$ is exponential in $n$. 

Our second result is that a Steiner spanner with near-optimal lightness
can be constructed in linear time, when the spread $\Delta$ is polynomial in $n$.

\begin{theorem}\label{thm:construction} 
	For any point set $P$ in $\mathbb{R}^2$ of spread $\Delta = O(n^{c})$, for any $c= O(1)$, 
	a Steiner $(1+\epsilon)$-spanner of lightness $O(\frac{\log \Delta}{\epsilon})$
	can be constructed in $O_{\epsilon}(n)$ time,  	
	where $O_{\epsilon}(.)$ hides a  factor of $\poly(\frac{1}{\epsilon})$.
\end{theorem}

\subsubsection*{Higher dimensions}

The lower bound of~\cite{Le20} 
states that any (non-Steiner) $(1+\epsilon)$-spanner must incur lightness $\Omega(\epsilon^{-d})$, for any $d = O(1)$. 
We show that, similarly to the 2-dimensional case, one can improve the lightness almost quadratically using Steiner points,
for point sets of bounded spread. 

\begin{theorem}\label{thm:any-dim} 
	For any $d \ge 3$, any point set $P$ in $\mathbb{R}^d$ of spread $\Delta$ 
	admits a Steiner $(1+\epsilon)$-spanner of lightness
	$\tilde{O}(\epsilon^{-(d+1)/2} + \epsilon^{-2}\log \Delta)$.
\end{theorem}

Interestingly, the dependence on the spread in the lightness bound provided by Theorem \ref{thm:any-dim} does not grow with the dimension, hence the improvement over the non-Steiner bound gets more significant as the dimension grows, provided of course that the spread is not too large. 

\paragraph{Follow-up Work} Recently Bhore and T\'{o}th~\cite{BT20} proved a lightness lower bound $\Omega(\epsilon^{-d/2})$ for Steiner spanners of a point set with spread $\poly(\frac{1}{\epsilon})$ for any $d\geq 2$. Their lower bound for $d= 2$ improved the lower bound in Le and Solomon~\cite{LS19} by a factor of $\log(\frac{1}{\epsilon})$. This implies that our upper bound in Theorem~\ref{thm:any-dim} is almost tight (up to a factor of $\tilde{O}(\epsilon^{-1/2})$)  for point set with spread $\poly(\frac{1}{\epsilon})$.  On the positive side, using \emph{directional} $(1+\epsilon)$-spanners, they showed that point set in $\mathbb{R}^2$ admits a Steiner $(1+\epsilon)$-spanner with lightness $O(\epsilon^{-1} \log n)$. This bound improved the lightness bound in Theorem~\ref{thm:bounded-spread} when the spread $\Delta$ is super-polynomial in $n$.  

We point in the proof overview below that using a standard technique~\cite{CLNS13,ES13}, we can easily construct a light Steiner $(1+\epsilon)$-spanner for $n$-point sets in $\mathbb{R}^2$ with lightness $O(\epsilon\log n)$. Using the same technique, the factor $\log \Delta$ in Theorem~\ref{thm:any-dim} can be replaced by $\log n$.

\subsection{Proof Overview} 

\subsubsection{Proof of Theorem~\ref{thm:bounded-spread}}
We partition the set of pairs of points into $m=O(\log \Delta)$ subsets $\{\mathcal{P}_1, \ldots, \mathcal{P}_{m}\}$ where $\mathcal{P}_i$ contains pairs of distances in $[2^{i-1}, 2^i)$. The objective is to show that one can preserve distances between all pairs in $\mathcal{P}_i$ to within a factor of $(1+\epsilon)$ using a Steiner spanner $S_i$ of weight $O(\frac{w(\mst)}{\epsilon})$; by taking the union of all such spanners, we obtain a Steiner spanner with the required lightness.  

Let $L_i = 2^i$.
A natural idea to preserve distances in $\mathcal{P}_i$ is to (a)  find an $(\epsilon L_i)$-net $N_i$\footnote{A subset of points $N\subseteq P$ is an $r$-net if every point in $P$ is within distance $r$ from (or {\em covered} by)
	some point in $N$  and pairwise distances between points in $N$ are larger than $r$.} of $P$ and (b) add to the spanner edges between any two net points $p$ and $q$ such that $(u,v) \in \mathcal{P}_i$ 
and there is a pair $(u,v)$ in $\mathcal{P}_i$ such that $u$ and $v$ are covered by (i.e., within distance $\epsilon L_i$ from) $p$ and $q$, respectively.
The stretch will be in check because $u$ and $v$ are at distance roughly $O(\epsilon) ||u,v||$ from their net points and thus, the {\em additive} stretch between $u$ and $v$ is $O(\epsilon) ||u,v||$. For lightness, we can show that the number of net points $|N_i| = O(\frac{w(\mst)}{\epsilon L_i})$, and using a (nontrivial) packing argument, there are about $\frac{N_i}{\epsilon}$ edges of length $O(L_i)$ added in step (b). Thus, the total weight of the spanner is $O(\frac{w(\mst)}{\epsilon^2})$, which is bigger than our aimed lightness bound by a factor of $\frac{1}{\epsilon}$.

To shave the factor of $\frac{1}{\epsilon}$, we employ two ideas. 
First, we take $N_i$ to be a $\sqrt{\epsilon}L_i$-net of $P$. In this way $|N_i| = O(\frac{w(\mst)}{\sqrt{\epsilon}})$. By applying the 2-dimensional Steiner spanner construction of~\cite{LS19} as a blackbox, we obtain a Steiner spanner with only $\frac{|N_i|}{\sqrt{\epsilon}}$ edges of length $O(L_i)$ that approximates distances between points in $N_i$; 
we have thus reduced the weight bound to $O(\frac{w(\mst)}{\epsilon})$, as required. The problem now is with the stretch guarantee: Preserving distances between the points in $N_i$ is no longer sufficient. Indeed, for every pair $(u,v) \in \mathcal{P}_i$, the distance between $u$ and $v$ to the nearest net points is $O(\sqrt{\epsilon}) ||u,v||$, hence the resulting stretch is $(1+O(\sqrt{\epsilon}))||u,v||$. We overcome this hurdle by introducing a novel construction of \emph{single-source spanners}, which generalize Steiner shallow-light trees of Solomon~\cite{Solomon15}. Specifically, we open the black-box of \cite{LS19} and observe that every time we want to preserve the distance from (some) Steiner point $s$ to a net-point $p\in N_i$, instead of connecting $s$ to $p$ by a straight line of weight $O(L_i)$, we can use a  single-source spanner (rooted at $s$) of \emph{weight $O(L_i)$} to preserve distances (up to a $(1+\epsilon)$ factor) from $s$ to \emph{every point} within distances $O(\sqrt{\epsilon} L_i)$ from $p$. As a result, our Steiner spanner  can preserve distances between any two points $(u,v)\in \mathcal{P}_i$ to within a factor of $(1+\eps)$, where $u \in B(p,O(\sqrt{\epsilon} L_i))$, $v \in B(q,O(\sqrt{\epsilon} L_i))$, and $p,q$ are their nearest net points.

\paragraph{Replacing $\log \Delta$ by $\log n$} A factor $\log(\Delta)$ incurred  in the lightness bound of Theorem~\ref{thm:bounded-spread} is because we have $m = O(\log \Delta)$ subsets $\{\mathcal{P}_1, \ldots, \mathcal{P}_m\}$. To reduce $m$ to $O(\log n)$, we observe that $\Delta \leq w(\mst)$ and that we can take every edge of weight at most $\frac{w(\mst)}{n^2}$ to the spanner since all such edges have total weight at most $O(w(\mst))$. Thus, the pairs of interest now have weights in the range $(\frac{w(\mst)}{n^2}, w(\mst)]$ and hence we set $m = \log n^2 = O(\log n)$.

\subsubsection{Proof of Theorem~\ref{thm:any-dim}}
By extending the construction in Theorem~\ref{thm:bounded-spread}, we can construct a Steiner spanner with lightness $\tilde{O}(\epsilon^{-(d+1)/2}\log \Delta)$ as follows: for each $\mathcal{P}_i$, we construct an $\epsilon L_i$-net $N_i$ and then apply the construction of~\cite{LS19} as a black box to obtain a Steiner spanner $S_i$ for $N_i$ with weight $\tilde{O}(\epsilon^{-(d-1)/2} |N_i| L_i)$. The stretch will be in check since $N_i$ is an $\epsilon L_i$-net.  Since $|N_i| = O(\frac{w(\mst)}{\epsilon L_i})$, $w(S_i) = \tilde{O}(\epsilon^{-(d+1)/2}w(\mst))$. The union of Steiner spanners $S_i$ for all $i \in [1,m]$ has weight $\tilde{O}(\epsilon^{-(d+1)/2}\log \Delta)w(\mst)$. Note when $d\geq 3$, we can not take $N_i$ as a $\sqrt{\epsilon}L_i$-net since the construction of single-source spanners with $O(L_i)$ weight  in the proof of Theorem~\ref{thm:bounded-spread} only works when $d=2$.

Most of our effort  is to further refine the result in a way that $\log \Delta$ term is multiplied only by $\eps^{-2}$ and not by the term that depends on $d$.  We first reduce to the problem of approximating distances between pairs (of endpoints) in a family of edge sets $\mathcal{E} = \{E_1,\ldots, E_m\}$ with $m = O(\log \Delta)$, where edges of $E_i$ have length (roughly) in the interval  $(\frac{1}{2\epsilon^{i}}, \frac{1}{\epsilon^i}]$ and edges in $E_1$ have length in $[1,\frac{1}{\epsilon})$. Let $L_i = \frac{1}{\epsilon^i}$ . For a technical reason, we will subdivide edges of $\mst$ by using a set of Steiner points $K$ so that each new edge has length in $(1/2,1]$.  

We construct a Steiner spanner for edges in $\mathcal{E}$ using a \emph{charging cover tree} $T$: $T$ has depth $m+1$,  level $i$ of $T$ is associated with an \emph{$\epsilon L_i$-cover}\footnote{A subset of points $N\subseteq P$ is an $r$-cover if every point in $P$ is within distance $r$ from some point in $N$.} of $P\cup K$, and  leaves (at level $0$) of $T$ are points in $P \cup K$. For each cover $N_i$ at level $i$ of $T$, we construct a graph $H_i$ where $V(H_i) = N_i$  and there is an edge $(u,v)$ between $u,v\in N_i$ if there is a corresponding edge in $E_i$ whose endpoints are covered by $u$ and $v$, respectively. Graph $H_i$ is used to distinguish between \emph{low degree points}, whose degree in $H_i$ is $O(\frac{1}{\epsilon})$, and \emph{high degree points}, whose degree in $H_i$ is $\Omega(\frac{1}{\epsilon})$. $T$ will have two \emph{charging properties}: (1) every point $p \in N_i$ has at least $\epsilon L_i$ \emph{uncharged descendants} and (2) at every level $i$, one can charge up to  $\frac{\epsilon L_i}{2}$ uncharged descendants of high degree points. Note that, once an uncharged point is \emph{charged} at level $i$, it will be marked as charged at higher levels; initially at level $0$, every point is \emph{uncharged}. 

We then use a charging cover tree $T$ to guide the Steiner spanner construction. Specifically, at level $i$, we add all edges incident to low degree points to the spanner and we can show that the total weight of all these edges over all levels is at most $O(\frac{ w(\mst) \log \Delta}{\epsilon^2})$. For high degree points, we apply the construction of Le and Solomon~\cite{LS19}  to obtain a Steiner spanner $S_i$,  and we charge the weight of $S_i$ to $\frac{L_i \epsilon}{2}$ uncharged descendants of each high degree point. This charging is possible by the charging property (2) of $T$. We then show that each point in $K\cup P$ is charged a weight at most $\tilde{O}(\epsilon^{-(d+1)/2})$. Thus, the total weight of the Steiner spanners (for high degree points) at all levels is $\tilde{O}(\epsilon^{-(d+1)/2})w(\mst)$.

Our construction of a charging cover tree is inspired by the construction of a hierarchy of clusters in the iterative clustering technique.  The technique was initially developed by  Chechick and Wulff-Nilsen~\cite{CW16} to construct light spanners for general graphs, and then was adapted to many other different settings~\cite{CW16,BLW17,BLW19,LS19,Le20}. Our construction is directly inspired by the construction of Borradaile et al.~\cite{BLW19} in the doubling dimension setting. However, our construction is much simpler. Specifically, we are able to decouple the Steiner spanner construction from the charging cover tree construction. We refer readers to Section~\ref{sec:high-dim-spanner} for more details.

\paragraph{Replacing $\log \Delta$ by $\log n$} Using the same technique pointed in the previous section, by taking every edge of weight at most $\frac{w(\mst)}{n^2}$ to the spanner, we are left with pairs of weights in range $(\frac{w(\mst)}{n^2}, w(\mst)]$. Thus, we can set  $m = \log n^2 = O(\log n)$ instead of $O(\log(\Delta))$.

\section{Preliminaries}\label{sec:preliminary}

Let $P$ be a point set of $n$ points in $\mathbb{R}^d$. We denote by $||p,q||$ the Euclidean distance between two points $p,q\in \mathbb{R}^d$.   Let $B(p,r) = \{x\in \mathbb{R}^d, ||p,x||\leq r\}$ be the ball of radius $r$ centered at $p$.  Given a point $p$ and a set of point $Q$ on the plane, we define the distance between $p$ and $Q$, denoted by $d(p,Q)$, to be $\inf_{x\in Q}||p,x||$.

An \emph{$r$-cover} of $P$ is a subset of points $N \subseteq P$ such that for every point $x \in P$, there is at least one point $p \in N$ such that $||p,x|| \leq r$; we say $x$ is \emph{covered} by $p$. When the value of $r$ is clear from the context, we simply call $N$ a cover of $P$.  A subset of point $N\subseteq P$ is called an \emph{$r$-net} if $N$ is an $r$-cover of $P$ and also an {\em $r$-packing} of $P$, i.e., for every two points $p\not= q \in N$, $||p,q|| > r$.

Let $G$ be a graph with weight function $w$ on the edges. We denote the vertex set and edge set of $G$ by $V(G)$ and $E(G)$, respectively.   Let $d_G(p,q)$ be the distance between two vertices $p,q$  of $G$. We denote by $G[X]$ the subgraph induced by a subset of vertices $X$.

$G$ is \emph{geometric}  in $\mathbb{R}^d$ if each vertex of $G$ corresponds to a point $p \in \mathbb{R}^d$ and for every edge $(p,q)$, $w(p,q) = ||p,q||$. In this case, we use points to refer to vertices of $G$. We say that a geometric graph $G$ is a \emph{$(1+\epsilon)$-spanner} of $P$ if $V(G) = P$ and for every two points $p\not=q  \in P$, $d_G(p,q)\leq (1+\epsilon)||p,q||$.  We say that $G$ is a Steiner \emph{$(1+\epsilon)$-spanner}  for $P$ if $P\subseteq V(G)$ and for every two points $p\not=q  \in P$, $d_G(p,q)\leq (1+\epsilon)||p,q||$. Points in $V(G)\setminus P$ are called \emph{Steiner points}.  Note that distances between Steiner points may not be preserved in a Steiner $(1+\epsilon)$-spanner.

\section{Steiner Spanners on the Plane}\label{sec:Plane-spanner}

We focus on constructing a Steiner spanner with good lightness; the fast construction is in Section~\ref{sec:fast-spanner-R2}. We will use the following geometric Steiner \emph{shallow-light tree} (SLT)  construction by Solomon~\cite{Solomon15}. 

\begin{lemma}\label{lm:STL} Let $L$ be a line segment of length $\sqrt{\epsilon}$ and $p$ be a point on the plane such that $d(p,L) = 1$. For any point set $X \in L$, there is a geometric graph $H$ of weight  $\Theta(1)$ such that $d_{H}(p,x) \leq (1+\epsilon)||p,x||$ for any point $x \in X$. 
\end{lemma}

We will use \emph{single-source spanners} (defined below) as a black box in our construction. 

\begin{definition}[Single-source spanners] Given a point $p$ (source), a set of points $X$ on the plane and a connected geometric graph $S_X$ spanning $X$, a \emph{single source $(1+\epsilon)$-spanner} w.r.t. $(p,X,S_X)$ is a graph $H$ such that for every $x \in X$: $||p,x|| \leq d_{H\cup S_X}(p,x) \leq (1+\epsilon) ||p,x||$.
\end{definition}
Our starting point is the construction of a single source spanner from a point $p$  to point set $X$ enclosed in a circle $C$ of radius $\sqrt{\epsilon}$ such that $d(p,C) = 1$. We show that, if $S_X$ approximately preserves the distances  between pairs of points in $X$ up to a $(1+g\epsilon)$ factor for any constant $g$, it is possible to construct a single-source spanner with weight $O(1)$. It is not so hard to see that if Steiner points are not allowed, a lower bound of weight $\Omega(\frac{1}{\sqrt{\epsilon}})$ holds here. 

\begin{lemma} \label{lm:single-source}Let $X$ be a set of points in a circle $C$ of radius $\sqrt{\epsilon}$ on the plane and a point $p$ of distance $1$ from $C$. Let $S_X$ be a $(1+g\epsilon)$-spanner of $X$ for any constant $g$. Then there is a single-source  $(1+13\epsilon)$-spanner $H$ w.r.t. $(p,X,S_X)$ of weight $O(1)$ when $g \ll \frac{1}{\epsilon}$.
\end{lemma}
\begin{proof}
	Let $c$ be a center of $C$.~W.l.o.g, we assume that $pc$ is parallel to $y$-axis. Let $Q$ be the axis-aligned smallest square bounding $C$. Observe that the side length of $Q$ is at mos $2\sqrt{\epsilon}$. Place a $\frac{2}{\sqrt{\epsilon}} \times \frac{2}{\sqrt{\epsilon}}$ grid $W$ on $Q$, so that every cell of $W$ is a square of side length $\epsilon$. Observe that:
	\begin{equation*}
	w(W) \leq \epsilon\cdot \frac{4}{\epsilon} = O(1)
	\end{equation*}
	
	We extend $W$ to $W_1$ by connecting an (arbitrary) corner of each grid cell to an arbitrary point of $X$ in the cell.  Observe that: $w(W_1) = O(W) = O(1)$. Let $P$ be the set of grid points on the side, say $L$, of $Q$ that is closer to $p$ (than the opposite side). We apply the construction in Lemma~\ref{lm:STL} to $p$ and $L$ to obtain a geometric graph $K$.  Let $H = W_1 \cup K$. Since $w(K) = O(1)$ by Lemma~\ref{lm:STL}, it holds that $w(H)  = O(1)$.

	It remains to show the stretch bound.  Let $x$ be any point of $X$ and $v$ be the point in the same cell with $x$ that is connected to a corner, say $z$ of grid $W$. We will show below that:	
	\begin{equation}\label{eq:gridpoint-stretch}
	d_{W\cup K}(z,p) \leq (1+3\epsilon)||p,z||
	\end{equation}
	
	If Equation~\ref{eq:gridpoint-stretch} holds, it would imply:
	\begin{equation*}
	\begin{split}
	d_{S_X\cup H}(x,p)& \leq d_{S_X}(x,v) + d_{H}(v,p) \leq (1+g\epsilon)\sqrt{2}\epsilon +  d_{H}(v,p)\\
	&\leq (1+g\epsilon)\sqrt{2}\epsilon + d_H(z,v) +  d_{H}(z,p) \overset{g\ll 1/\epsilon}{\leq} 2\sqrt{2} \epsilon  +  \sqrt{2}\epsilon + d_{W\cup K}(z,p)\\
	&  \overset{\text{Eq.~\ref{eq:gridpoint-stretch}}}{\leq}  (1+3\epsilon)||p,z|| + 3\sqrt{2}\epsilon \overset{||x,z||\leq \sqrt{2}\epsilon}{\leq}  (1+3\epsilon)(||p,x||  + \sqrt{2} \epsilon) + 3\sqrt{2}\epsilon\\
	&\leq  (1+3\epsilon)||p,x||  + 7\sqrt{2}\epsilon \leq (1+13\epsilon) ||p,x|| \qquad \mbox{since } ||p,x||\geq 1
	\end{split}
	\end{equation*}
	Thus, it remains to prove Equation~\ref{eq:gridpoint-stretch}. To this end, let $y$ be the projection of $z$ on $L$. (Point $y$ is also a grid point; see Figure~\ref{fig:sssp})  Let $y',z'$ be projections of $y$ and $z$ on the line containing $pc$, respectively. Let $u$ be the intersection of $pz$ and $yy'$. Observe that $||p,y|| \leq (1+\epsilon) \leq (1+\epsilon) ||p,u||$. Thus, 
	\begin{equation}\label{eq:pyz-vs-pz}
	||p,y|| +  ||y,z|| ~\leq~  (1+\epsilon) ||p,u||  + ||y,z||\leq (1+\epsilon) ||p,u||  + ||u,z|| \leq (1+\epsilon)||p,z||
	\end{equation}
	Since $d_{W}(z,y) = ||z,y||$ and $ d_K(y,p) = (1+\epsilon)||p,y||$ by Lemma~\ref{lm:STL}, we have:
	\begin{equation}
	d_{W\cup K}(z,p) ~\leq~ (1+\epsilon) (||z,y|| +||p,y||)~\overset{Eq.~\ref{eq:pyz-vs-pz}}{\leq}~ (1+\epsilon)( 1+\epsilon)||p,z|| ~\leq~ (1+3\epsilon)||p,z|| 
	\end{equation}
	which implies Equation~\ref{eq:gridpoint-stretch}.
\end{proof}
\begin{figure}{r}
	\centering{\includegraphics[width=.8\textwidth]{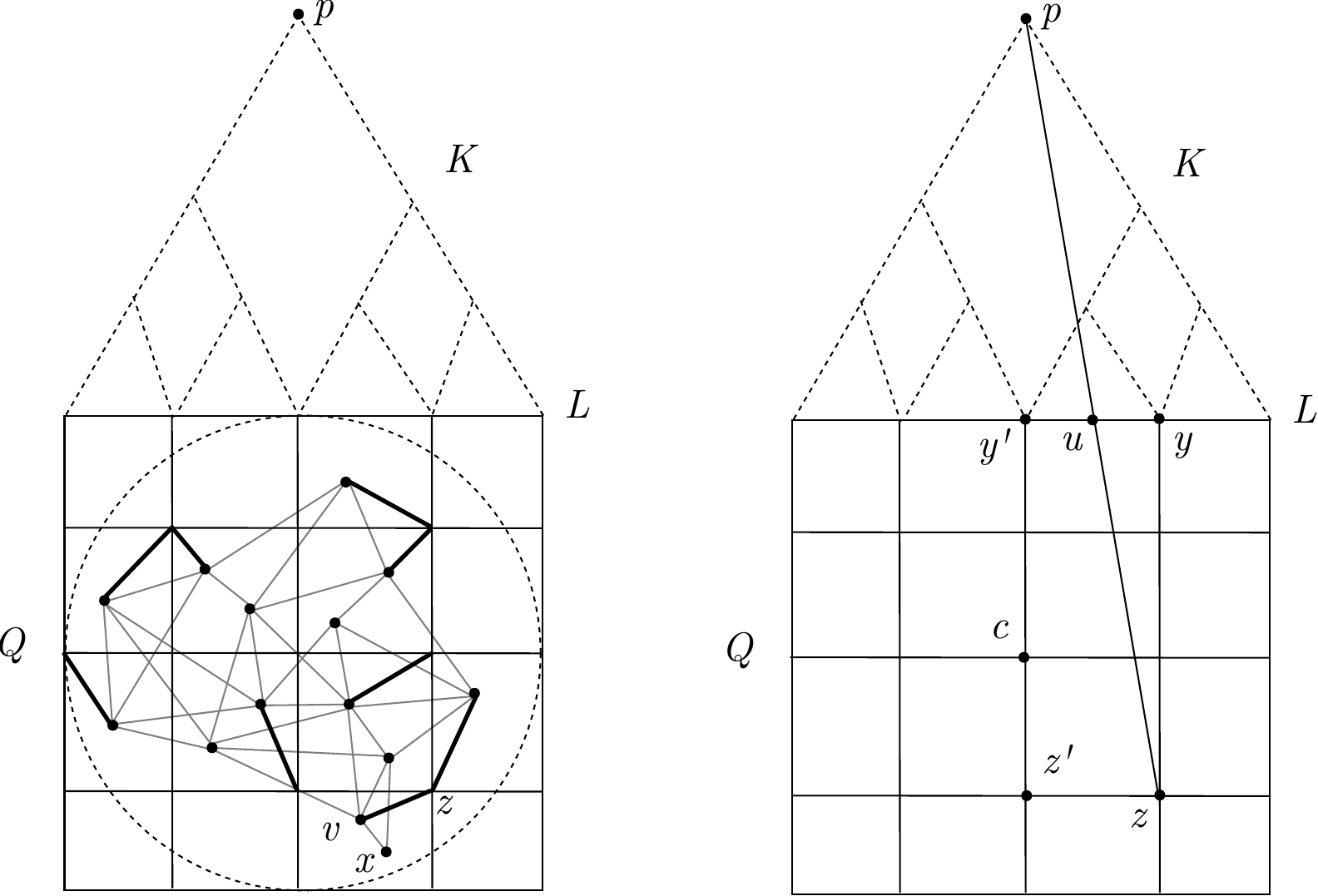}}
	\caption{\label{fig:sssp}\small \it 
		(Left)	A single source spanner from $p$ to a set of points enclosed by a circle of radius $\sqrt{\epsilon}$. One point in each non-empty cell is connected to a corner by a thick edge. (Right) An illustration for analyzing the stretch of $pz$.
	}
	\vspace{-20pt}
\end{figure}
We obtain the following corollary of Lemma~\ref{lm:single-source}.

\begin{corollary} \label{cor:sssp-any-scale}Let $X$ be a set of points in a circle $C$ of radius $\sqrt{\epsilon}L$ on the plane and a point $p$ of distance $L/h$ from $C$ for some constant $h \geq 1$. Let $S_X$ be a $(1+g\epsilon)$-spanner of $X$ for any constant $g$. Then there is a single-source  $(1+13\epsilon)$-spanner $H$ w.r.t. $(p,X,S_X)$ of weight $O(hL)$ when $g \ll \frac{1}{\epsilon}$.
\end{corollary}
\begin{proof}
	We scale the space by $L/h$. In the scaled space, $C$ has radius $h\sqrt{\epsilon}$ and $d(p,C) = 1$. Let $C_1,C_2,\ldots, C_m$, where $m = O(h^2)$, be circles of radius $\sqrt{\epsilon}$ covering $C$; such a set of circles can be constructed greedily. We apply Lemma~\ref{lm:single-source} to $p$ and each $C_i$ to construct a single-source $(1+13\epsilon)$-spanner $H_i$ from $p$ to each $C_i$. The final spanner is $H = \cup_{i=1}^m H_i$ that has total weight $O(h^2)$ in the scaled metric. Thus, in the original metric, $w(H) = O(h^2L/h) = O(hL)$.
\end{proof}
We are now ready to prove Theorem~\ref{thm:bounded-spread}.

\begin{proof}[Proof of Theorem~\ref{thm:bounded-spread}]	Assume that the minimum pairwise distance is $1$.   Let $\mathcal{P} = {P \choose 2}$ be all pairs of points. Partition $\mathcal{P}$ into $O(\log \Delta)$  sets $\mathcal{P}_1, \mathcal{P}_2, \ldots ,\mathcal{P}_{\lceil \log \Delta\rceil}$ where $\mathcal{P}_i$ is the set of pairs $(x,y)$ such that $||x,y|| \in [2^{i-1}, 2^{i})$.
	
	For a fixed $i$, we claim that there is a geometric graph $H_i$ such that for every two distinct points $(x,y) \in \mathcal{P}_i$,  $d_{H_1\cup \ldots \cup H_i}(x,y) \leq (1+\epsilon)||x,y||$ and that $w(H_i) = O(\frac{1}{\epsilon})w(\mst)$. Thus, $H_1\cup \ldots \cup H_{\lceil \log \Delta\rceil}$ is a Steiner spanner with weight $O(\frac{\log\Delta}{\epsilon})w(\mst)$.

	We now focus on constructing $H_i$. Let $S_{i-1} = H_1\cup H_2\ldots \cup H_{i-1}$. We will construct a spanner with stretch $(1+c\epsilon)$ for some constant $c$.  By induction, we can assume that:
	\begin{equation}\label{eq:stretch-induction}
	d_{S_{i-1}}(p,q) \leq (1+c\epsilon) ||p,q||
	\end{equation}
	for any pair $(p,q) \in \mathcal{P}_1\cup \ldots \cup \mathcal{P}_{i-1}$.
	
	Let $L_i = 2^{i}$ and $N_i$ be a $(\sqrt{\epsilon}L_i)$-net of $P$. For each point $x$, let $N_i(x)$ be the net point that \emph{covers} $x$: the distance from $x$ to $N_i(x)$ is at most $\sqrt{\epsilon} L_i$.
	
	\begin{claim}\label{clm:net-mst} $|N_i| = O(\frac{w(\mst)}{L_i\sqrt{\epsilon}})$.
	\end{claim}
	\begin{proof}
		Consider the circle $B(p,\sqrt{\epsilon}L_i)$ centered at $p$; $B(p,\sqrt{\epsilon}L_i)$ contains a segment of length  $\Omega(\sqrt{\epsilon} L_i)$ of the MST, which is not contained in any other circle. Thus, the claim holds.
	\end{proof}
	
	Next, we consider the smallest axis-aligned square $Q$ bounding the point set. In the following, we divide $Q$  into a set of (overlapping) sub-squares $\mathcal{B}$ of side length $\Theta(L_i)$ each.  This way, for any pair $(x,y) \in \mathcal{P}_i$ (of distance at most $L_i$), there is a sub-square entirely containing $N_i(x), N_i(y)$, and the balls of radius $O(\sqrt{\epsilon} L_i)$ around the two net points. 
	
	\begin{quote}
		\textbf{Constructing $\mathcal{B}$} We first divide $Q$ into subsquares of side length $5 L_i$ each\footnote{We assume that the side length of $Q$ is divisible by $5L_i$; otherwise, we can extend $Q$ in such a way.}. For each subsquare $B$, we extend its borders equally to four directions by an amount of $2L_i$ in each direction. After this extension, $B$ has side length $9L_i$. 
	\end{quote}
	
	\begin{figure}
		\centering{\includegraphics[width=.8\textwidth]{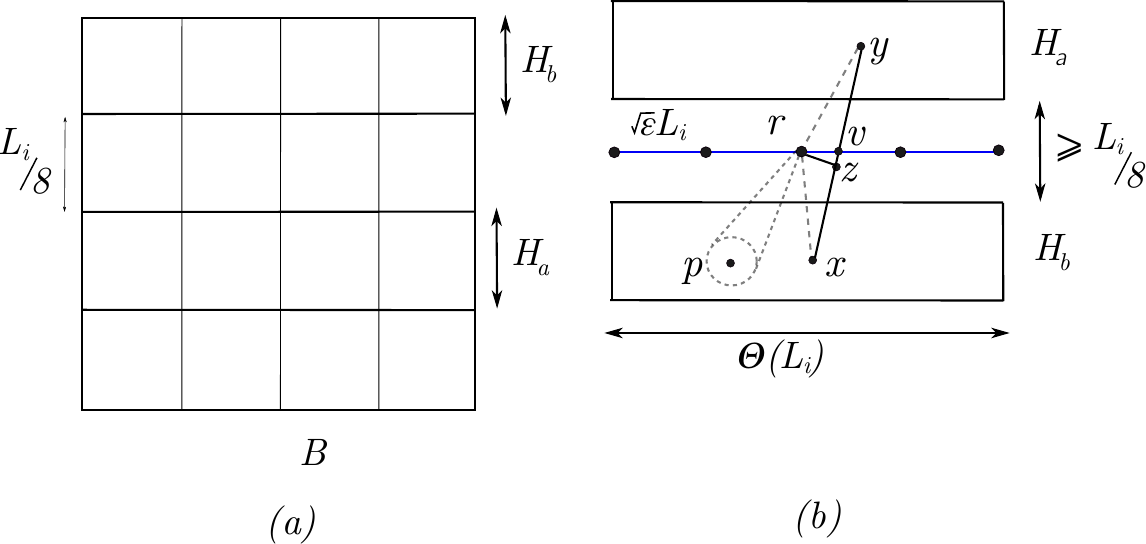}}
		\caption{\label{fig:banding}\small \it 
			(a)	Square $B$ is divided into $O(1)$ horizontal and vertical bands of width $L_i/8$ each. (b) The Steiner spanner construction for two non-adjacent horizontal bands $H_a$ and $H_b$. The dashed cone represents a single-source spanner from $r$ to circle $B(p,\sqrt{\epsilon}L_i)$.
		}
	\end{figure}
	\begin{claim}\label{clm:point-bounded-squares} Every point in $Q$ belongs to at most $4$ subsquares in $\mathcal{B}$. Furthermore, for each pair $(x,y) \in \mathcal{P}_i$, there is a subsquare $B\in \mathcal{B}$ such that $B(N_i(x), \sqrt{\epsi}L_i),B(N_i(y), \sqrt{\epsi}L_i)$ are entirely contained in $B$.
	\end{claim}
	\begin{proof}
		Let $B$ be a subsquare in $\mathcal{B}$ containing one of the endpoints of $(x,y)$, say $x$, before extension. Then, after extension, $B$ will contain both $x,y$ since $||x,y|| \leq L_i$, and furthermore, $x$ and $y$ are at least $L_i$ away from the boundary since we extended $B$ by $2L_i$ in each direction. Thus, points in $B(N_i(x), \sqrt{\epsi}L_i)$ ($B(N_i(y), \sqrt{\epsi}L_i)$) will be at most $2\sqrt{\epsilon} L_i < L_i$ from $x$ ($y$) when $\epsilon \ll 1$.
	\end{proof}
	
	Consider a subsquare $B \in \mathcal{B}$. Let $N_B = N_i\cap B$.  By abusing notation, we denote by  $\mathcal{P}_i \cap B$ all the pairs in $B$ of $\mathcal{P}_i$. We will show that:
	
	\begin{claim} \label{clm:box-spanener}  There is a Steiner spanner $S_B$ of weight at most $O(|N_B| L_i/\sqrt{\epsilon})$ such that for any pair of points $(x,y) \in  \mathcal{P}_i\cap  B$, it holds that:
		\begin{equation*}
		d_{S_B \cup S_{i-1}}(x,y) \leq (1+c\epsilon)||x,y||
		\end{equation*}
		for some big enough constant $c$.
	\end{claim}
	\begin{proof}		
		We divide $B$ into $O(1)$ horizontal (vertical) \emph{bands} of length (width)  $L_i/8$ so that for any two points $x,y \in \mathcal{P}_i\cap B$,  $N_i(x)$ and $N_i(y)$ are in two non-adjacent horizontal bands and/or vertical bands (see Figure~\ref{fig:banding}).  These bands exist since
		\begin{equation*}
		||N_i(x), N_i(y)|| ~\geq~ ||x,y||-2\sqrt{\epsilon} L_i ~\geq~ L_i/2-2\sqrt{\epsilon} L_i~\geq~ L_i/4  
		\end{equation*}

		Now for each pair of non-adjacent horizontal bands $H_a$ and $H_b$, $d(H_a,H_b) \geq \frac{L_i}{8}$. Draw a bisecting segment (touching two sides of $B$) between  $H_a$ and $H_b$ and place $O(\frac{1}{\sqrt{\epsilon}})$ equally-spaced Steiner points, say $R$, on the bisecting line in a way that the distance between any two nearby Steiner points is $L_i\sqrt{\epsilon}$ (see Figure~\ref{fig:banding}(a)). For each point $r \in R$, and each net point $p_i \in N_B\cap (H_a\cup H_b)$, we apply the construction of Corollary~\ref{cor:sssp-any-scale} to $r$, the set of endpoints of $\mathcal{P}_i$ inclosed in circle $B(p_i, \sqrt{\epsilon}L_i)$ and $S_{i-1}$; let $S_{a,b}(p_i)$ be the obtained geometric graph (see Figure~\ref{fig:banding}(b)). Note that $d(r, B(p_i, \sqrt{\epsilon}L_i)) = \Omega(L_i)$. Thus, by Corollary~\ref{cor:sssp-any-scale}, $w(S_{a,b}(p_i)) = O(L_i)$ and that:
		\begin{equation}\label{eq:band-stretch}
		d_{S_{i-1}\cup S_{a,b}(p_i)}(r, q) \leq (1+13\epsilon) ||r,q||
		\end{equation}
		for any $q \in B(p_i, \sqrt{\epsilon}L_i)\cap P$. 		Let $S_{a,b}(r) = \cup_{p_i \in N_B\cap (H_a\cup H_b)} S_{a,b}(p_i)$. It holds that $$w(S_{a,b}(r)) \leq O(L_i)|N_B\cap (H_a\cup H_b)|.$$ Let $S_{a,b} = \cup_{r\in R}S_{a,b}(r)$. Then, we have:		
		\begin{equation}
		w(S_{a,b}) \leq  O(L_i \cdot |R| \cdot |N_B\cap (H_a\cup H_b)|) = O(\frac{L_i}{\sqrt{\epsilon}}|N_B\cap (H_a\cup H_b)|)
		\end{equation}
		We apply the same construction for every pair of non-adjacent vertical bands.  We then let $S_B$ be the union of all $S_{a,b}$ for every pair of non-adjacent horizontal/vertical bands $H_a,H_b$. It holds that:
		\begin{equation}
		w(S_B)  = O(L_i \cdot |R| \cdot |N_B|) = O(\frac{L_i |N_B|}{\sqrt{\epsilon}})
		\end{equation}
		since there are only $O(1)$ pairs of bands. 	To bound the stretch, let $(x,y)$ be a pair in $\mathcal{P}_i$ whose endpoints are in $ B$. W.l.o.g, assume that $H_a,H_b$ are two non-adjacent horizontal bands that contain $N_i(x)$ and $N_i(y)$, respectively. Let $v$ be the intersection of segment $xy$ and the bisecting line $L$ of $H_a,H_b$ (see Figure~\ref{fig:banding}(b)). Let $r \in R$ be the closest Steiner point to $v$ in $L$ and $z$ be the projection of $r$ on $xy$. Observe that $||z,x||, ||z,y|| \geq \frac{L_i}{16} - ||z,v|| \geq L_i/16 - \sqrt{\epsilon}L_i \geq L_i/32$ when $\epsilon \ll 1$. We have:
		
		\begin{equation*}
		\begin{split}
		||r,x|| + ||r,y|| &= \sqrt{||x,z||^2 + ||r,z||^2} +  \sqrt{||y,z||^2 + ||r,z||^2}\\
		&\leq  ||x,z||\sqrt{1 + \frac{\epsilon L^2_i}{||x,z||^2}} +  ||y,z||\sqrt{1 +  \frac{\epsilon L_i^2}{||y,z||^2}} \qquad \mbox{since } ||r,z|| \leq \sqrt{\epsilon} L_i\\
		&\leq ||x,z||\sqrt{1 + 1024\epsilon} + ||y,z||\sqrt{1 + 1024\epsilon} \qquad \mbox{since } ||x,z||, ||y,z|| \geq L_i/32\\
		&\leq ||x,z|| (1 + 512\epsilon) +  ||y,z|| (1 + 512\epsilon) = (1+512\epsilon)||x,y||
		\end{split}		
		\end{equation*}
		Thus, by Equation~\ref{eq:band-stretch}, we have:
		\begin{equation*}
		d_{S_B \cup S_{i-1}}(x,y) \leq (1+13\epsilon)(||r,x|| + ||r,y|| ) =   (1+13\epsilon)(1+512\epsilon)||x,y|| = (1+O(\epsilon))||x,y||
		\end{equation*}
		Thus, the stretch is $(1+c\epsilon)$ for a sufficiently big constant $c$. 
	\end{proof} 
	
	Let $H_i = \cup_{B \in \mathcal{B}}S_B$. Since each net point belongs to at most $4$ subsquares in $\mathcal{B}$ by Claim~\ref{clm:point-bounded-squares}, and  $w(H_i) \leq O(\frac{|N_i|L_i}{\sqrt{\epsilon}})$ by Claim~\ref{clm:box-spanener}, it holds that:
	\begin{equation}
	w(H_i) = O(\frac{L_i}{\sqrt{\epsilon}} \frac{w(\mst)}{L_i \sqrt{\epsilon}}) = O(\frac{w(\mst)}{\epsilon})
	\end{equation}
	by Claim~\ref{clm:net-mst} as desired.
\end{proof}

\section{A Linear Time Construction}\label{sec:fast-spanner-R2}

In this section, we assume that $\Delta = O(n^c)$ for some constant $c$, and $\epsilon$ is a constant.  We use the same model of computation used by Chan~\cite{Chan08}: the real-RAM model with $\Theta(\log n)$ word size and floor function.  We will use $O_{\epsilon}$ notation to hide a polynomial factor of $\frac{1}{\epsilon}$.  Chan~\cite{Chan08} showed that:

\begin{theorem}[Step 4 in~\cite{Chan08}]\label{thm:chan} Given a poin set $P \in \mathbb{R}^d$ with spread $\Delta = O(n^{c})$  for \emph{constant} $d$ and  $c$, a $(1+\epsilon)$-spanner of $P$ can be constructed in $O_{\epsilon}(P)$ time.
\end{theorem}

We will use a construction of an $r$-net for a point set $P$ for any $r$ in  time $O(n)$.  Such a construction was implicit in the work of Har-Peled~\cite{HarPeled04} which was made explicit by Har-Peled and Raichel (Lemma 2.3~\cite{HR15}). 

\begin{lemma}[Lemma 2.3 and Corollary 2.4~\cite{HR15}]\label{lm:net-construction}
	Given $r \geq 1$ and an $n$-point set $P$ in $\mathbb{R}^d$, an $r$-net  $N$ of $P$ can be constructed in $O(n)$ time. Furthermore, for each net point $p$, one can compute all the points covered by $p$ in total $O(n)$ time.
\end{lemma}

We first show that the spanner in Corollary~\ref{cor:sssp-any-scale} can be implemented in $O_{\epsilon}(|X|)$ time.  

\begin{claim}\label{clm:fast-sssp} The single-source spanner $H$ in Corollary~\ref{cor:sssp-any-scale} can be found in $O_{\epsilon}(|X|)$ time.
\end{claim}
\begin{proof} First, we observe that the single source spanner in Lemma~\ref{lm:single-source} can be constructed in time $O_{\epsilon}(|X|)$. This is because the grid $W$ has size $O_{\epsilon}(1)$ and the SLT tree from $p$ to (a set of $O(\frac{1}{\sqrt{\epsilon}})$ grid points on) $L$ can be constructed in $O_{\epsilon}(1)$ time.  The single-source spanner in Corollary~\ref{cor:sssp-any-scale} uses a constant number of constructions construction in Lemma~\ref{lm:single-source}.  Thus, the total running time is $O_{\epsilon}(|X|)$.
\end{proof}

\begin{proof}[Proof of Theorem~\ref{thm:construction}]
	
	Our implementation will follow the construction in Section~\ref{sec:Plane-spanner} ; we will reuse notation in that section as well. Let $K$ be a $(1+\epsilon)$-spanner $H$ for $P$ constructed in time $O(n)$ by Theorem~\ref{thm:chan}. In our fast construction algorithm, instead of considering all pairs of points $\mathcal{P} = { P \choose 2}$, we only consider the pairs corresponding to edges of $H$; there are $O(n)$ such pairs. Our algorithm has  four steps:
	
	\begin{itemize}
		\item \textbf{Step 1} Partition pairs of endpoints in $E(H)$ into at most $O(\log(\Delta))$ sets $\mathcal{P}_1, \ldots, \mathcal{P}_{\lceil \log \Delta \rceil}$ where $\mathcal{P}_i$ is the set of pairs $(x,y)$ such that $||x,y|| \in [2^{i-1}, 2^{i})$. Let $L_i = 2^i$. This step can be implemented in time $O(|E(H)|) = O_{\epsilon}(n)$.	The following steps are applied to each $i \in [1,\lceil \log \Delta \rceil]$. Let $P_i$ be the set of endpoints of $\mathcal{P}_i$. We observe that:	
		\begin{equation}\label{eq:sum-Pi-size}
		\sum_{i=1}^n|P_i| ~\leq~ 2\sum_{i=1}^n |\mathcal{P}_i| ~=~ 2|E(H)| ~=~ O(n)
		\end{equation}
		\item  \textbf{Step 2} Construct a $(\sqrt{\epsilon}L_i)$-net $N_i$ for $P_i$ in $O(|P_i|)$ time using the algorithm in Lemma~\ref{lm:net-construction}.
		
		\item \textbf{Step 3} Compute a bounding square $Q$ and divide it into (overlapping) subsquares of length $\Theta(L_i)$ each. Let $\mathcal{B}$ be the set of subsquares that contain at least one point participating in the pairs $\mathcal{P}_i$. Since there are only $O(|P_i|)$ non-empty subsquares, $\mathcal{B}$ can be computed in time $O(|P_i|)$ by iterating over each point and check (in $O(1)$ time) which subsquare the point falls into. Here we use the fact that each floor operation takes $O(1)$ time.
		
		\item \textbf{Step 4} For each subsquare $B \in \mathcal{B}$, we divide it into $O(1)$ horizontal bands and vertical bands of length $L_i/8$. For each pair of non-adjacent (horizontal) bands $H_a,H_b$, construct a set of $O(\frac{1}{\sqrt{\epsilon}})$ Steiner points on the bisecting line between $H_a$ and $H_b$ as in Section~\ref{sec:Plane-spanner}.  For each Steiner point $r$ and each net point $p \in N_i\cap (H_a\cap H_b)$, we apply Corollary~\ref{cor:sssp-any-scale} to construct a single source spanner from $r$ to a set of points $B(p, \sqrt{\epsilon} L_i)\cap P_i$; this step can be implemented in time $O_{\epsilon}(|B(p, \sqrt{\epsilon} L_i)\cap P_i|)$ by Claim~\ref{clm:fast-sssp}. By Claim~\ref{clm:point-bounded-squares}, the construction in this step can be implemented in $O_{\epsilon}(|P_i|)$ time. Our final spanner is the union of all single source spanners in all subsquares in $\mathcal{B}$.
	\end{itemize} 
	
	The running time needed to implement Steps 2 to 4 is $O_{\epsilon}(\sum_{i=1}|P_i|) = O_{\epsilon}(n)$ by Equation~\ref{eq:sum-Pi-size}. The same analysis in Section~\ref{sec:Plane-spanner} gives $O(\frac{\log \Delta}{\epsilon})$ lightness. For stretch, we observe that the stretch of the spanner for each edge of $H$ is $(1+O(\epsilon))$.  Thus, the stretch for every pair of points in $P$ is  $(1+\epsilon)(1+O(\epsilon))  = (1+O(\epsilon))$. We can recover stretch $(1+\epsilon')$ by setting $\epsilon' = \frac{\epsilon}{c}$ where $c$ is the constant behind big-O. 
\end{proof}

\section{ Steiner Spanners in High Dimension}\label{sec:high-dim-spanner}

In this section, we a light Steiner spanner for a point set $P\in \mathbb{R}^d$ with spread $\Delta$ as in Theorem~\ref{thm:any-dim}. We rescale the metric so that every edge in ${P\choose 2}$ has weight at least $\frac{1}{\epsilon}$. Let $\mst$ be the minimum spanning tree of $P$.  We subdivide each $\mst$ edge of length $> 1$, by placing Steiner points greedily, in a way that each new edge has length \emph{at least $1/2$ and at most $1$}. Let $K$ be the set of Steiner points.  We observe that:
\begin{equation}\label{eq:mst-vs-points}
w(\mst) = \Theta(|P| + |K|)
\end{equation}
Let $\delta > 1$ be some parameter and $L_i = \frac{\delta}{\epsilon^i}$. Let $\mathcal{E}_{\delta} = \{E_1,\ldots, E_{m}\}$ be the set of edges such that 
\begin{equation}\label{eq:def-E}
E_i = \{e |  e \in {P \choose 2} \wedge w(e) \in (L_i/2, L_i]\}
\end{equation}
where $m = \lceil \log_{\frac{1}{\epsilon}} (\Delta/(\epsilon\delta)) \rceil ~\leq~ \lceil \log_{\frac{1}{\epsilon}}\Delta \rceil + 1$.  If an edge $e \in E_i$ for some $E_i \in \mathcal{E}_{\delta}$, we will abuse notation by saying that $e \in \mathcal{E}_{\delta}$. The main focus of this section is to show that:

\begin{lemma}\label{lm:lightSP-cover-tree}
	There is a Steiner spanner that preserves distances between the endpoints of edges in $\mathcal{E}_{\delta}$ with weight $\tilde{O}(\epsilon^{-(d+1)/2} + (\delta + \epsilon^{-2}) \log_{\frac{1}{\epsilon}} \Delta) w(\mst)$.
\end{lemma}

We will show below that Lemma~\ref{lm:lightSP-cover-tree} implies Theorem~\ref{thm:any-dim}.

\begin{proof}[Proof of Theorem~\ref{thm:any-dim}] 
	We assume that $\frac{1}{\epsilon}$ is a power of $2$. We partition the interval $[1,\epsilon)$ into $J = \log_2(\frac{1}{\epsilon})$ intervals $[1,2), \ldots, [2^{J-1}, 2^J)$. For each fixed $j \in [1,J]$, let $\delta_i = 2^j$, and $\mathcal{E}_{\delta_j}$ be the set of edges with $\delta = \delta_j$ in the definition of $\mathcal{E}_{\delta}$. Recall that we scale the metric so that every edge in ${P\choose 2}$ has weight at least $\frac{1}{\epsilon}$. Thus, 
	${P\choose 2} = \cup_{j=1}^J \mathcal{E}_{\delta_j}$
	
	Observe that $\delta_j \leq \frac{1}{\epsilon}$ for all $j\in [1,J]$.  By Lemma~\ref{lm:lightSP-cover-tree}, there exists a Steiner spanner $S_j$ with weight $\tilde{O} (\epsilon^{-(d+1)/2} + (\epsilon^{-2}\log_{\frac{1}{\epsilon}} \Delta)w(\mst)$ preseving distances between endpoints of edges in $\mathcal{E}_{\delta_j}$ up to a $(1+\epsilon)$ factor.	Then $S = \cup_{i=1}^J S_j$ is a Steiner $(1+\epsilon)$-spanner with weight $\tilde{O}(\epsilon^{-(d+1)/2} + \epsilon^{-2}\log \Delta)w(\mst)$.
\end{proof}

We now focus on constructing a Steiner spanner in Lemma~\ref{lm:lightSP-cover-tree}.   We will use the following Steiner spanner construction as a black box.

\begin{theorem}[Theorem 1.3~\cite{LS19}]\label{thm:LS19-high-dim}
	For a given point set $P$, there is a Steiner $(1+\epsilon)$-spanner, denoted by $\stp(P)$, with $\tilde{O}(\epsilon^{-(d-1)/2}) |P|$ edges that preserves pairwise distances of points in $P$ up to a $(1+\epsilon)$ factor.
\end{theorem}

We will rely on a \emph{cover tree} to construct a Steiner spanner.   Let $c$ be a sufficiently big constant chosen later ($c = 20$).

\begin{definition}[Cover tree] A cover $T$ for point set $P\cup K$ with $(m+1)$ levels has each node associated with a point of $P$ such that (a) level-$0$ of $T$ is the point set $P\cup K$,  (b) level-$i$ of $T$ is associated with a $(c \epsilon L_i)$-cover $N_i$ of $P$ and (c) $N_m \subseteq N_{m-1} \subseteq \ldots \subseteq N_0$. 
\end{definition}

A point $p$ may appear in many levels of a cover tree $T$. To avoid confusion, we  denote by $(p,i)$ the copy of $p$ at level $i$, and we still call $(p,i)$ a \emph{point} of $P\cup K$. For each point $(p,i)$, we denote by $\child(p,i)$ and $\desc(p,i)$ the set of children and  descendants of $(p,i)$ in $T$, respectively.   Note that $\desc(p,i)$ includes $(p,i)$. 

We will construct the Steiner spanner level by level, starting from level $1$. At every level, we will add a certain set of edges to $E_{sp}$. We then charge the weight of a subset of  the edges to a subset of \emph{uncharged points} of $P\cup K$; initially, every point of $P\cup K$ is \emph{uncharged}. To decide which uncharged points we will charge to at level $i$, we consider a geometric graph $H_i$ where $V(H_i) = N_i$ and there is an edge between two points $(p,i)\not= (q,i)$ (of weight $||p,q||$) in $H_i$ if there exists at least one $e \in E_i$ between two descendants of $(p,i)$ and $(q,i)$, respectively.  We say a cover point $(p,i)$ has \emph{high degree} if its degree in $H_i$ is at least $\frac{4c}{\epsilon}$.
At level $i$, we only charge to uncharged points which are descendants of high degree cover points. The intuition is that high cover points have many descendants. This leads to a notion of a \emph{charging cover tree}.  We call a cover tree a charging cover tree for $\mathcal{E}_{\delta}$ if for all level $i\geq 1$, 

\begin{itemize}
	\item [(1)] Each point $(p,i)$ has at least $\epsilon L_i$ descendants that are uncharged at level less than $i$. 
	\item[(2)] Up to $\epsilon L_i/2$ uncharged descendant of each high-degree  cover point $(p,i)$ can be charged at level $i$. No descendant of low-degree points is charged at level $i$. 
\end{itemize}

We show how to construct a charging cover tree in Appendix~\ref{app:charging-cover}.  We now show that given a charging cover tree, we can construct a Steiner spanner with the lightness bound in Lemma~\ref{lm:lightSP-cover-tree}. Let $T$ be such a charging cover tree. We define a set of edges $E_T$ as follows:
\begin{equation}\label{eq:def-ET}
E_T = \{(p,q) | (p \in \child(q,i) \vee q \in \child(p,i)) ~ \mbox{ for some $i$}\}
\end{equation}
We abuse notation by denoting $E_T$ the graph induced by the set of edges in $E_T$. 

\begin{claim}\label{clm:ET}
	$w(E_T) = O(\delta + \epsilon^{-1}\log_{\frac{1}{\epsilon}} \Delta) w(\mst))$ and for any $p$ and every $x\not= y \in \desc(p,i)$, $d_{E_T}(x,y) \leq 4c\epsilon L_i$.
\end{claim}
\begin{proof}
	Edges in $E_T$ can be partitioned according to levels where an edge is at level $i$ if it connects a point $(p,i)$ and its parent.  
	Observe that at level $0$, the total edge weight is at most $c\delta$ times the number of points and hence, the total weight is $O(\delta |K\cup P|) = O(\delta w(\mst))$ by Equation~\ref{eq:mst-vs-points}. At higher level, we observe that  the total weight of edges of $E_T$ at level $i\geq 1$ is at most $L_i|N_i|$, and that $N_i \leq \frac{|K\cup P|}{|\epsilon L_i|}$ since each point $(p,i)$ has $|\desc(p,i)|\geq \epsilon L_i$ by property (1) of the charging tree. Thus, the total weight of edges $E_t$ at level at least $1$ is at most:	
	\begin{equation*}
	\sum_{i=1}^{m} \frac{|K\cup P|}{\epsilon} = m \frac{|K\cup P|}{\epsilon} = O(\epsilon^{-1}\log_{\frac{1}{\epsilon}} \Delta |K\cup P|) = O(\epsilon^{-1}\log_{\frac{1}{\epsilon}} \Delta) w(\mst)
	\end{equation*}
	This implies the weight bound of $E_T$. We now bound the distance between $x\not= y \in \desc(p,i)$. Let $x =v_0, v_1, \ldots, v_k = p$ be the (unique) path from $x$ to $p$. By construction, $||v_{i-1},v_i|| \leq \frac{v_{i}v_{i+1}}{\epsilon}$ for $i \in [1,k-1]$. This implies:
	\begin{equation*}
	d_{E_T}(x,p)  = ||p,v_{k-1}|| \sum_{i=0}^{k-1}\eps^i \leq \frac{||p,v_{k-1}||}{1-\epsilon} \leq 2||p,v_{k-1}|| = 2c\epsilon L_i
	\end{equation*}
	when $\epsilon \leq\frac{1}{2}$. Similarly,  $d_{E_T}(y,p) \leq 2c\epsilon L_i$ and hence $d_{E_T}(x,y) \leq 4c\epsilon L_i$.
\end{proof}
Claim~\ref{clm:ET} implies that the descendants of any level $i$ node in a charging cover tree form a subgraph of diameter at most $O(c \epsilon L_i)$. Let $E_{sp}$ be the set of edges that will be our final spanner. Initially, $E_{sp} = E_T \cup \mst$. We will abuse notation by denoting $E_{sp}$ the graph induced by edge set $E_{sp}$. 

\subsection{Spanner construction at level $i$}  Recall that $H_i$ is a  geometric graph where $V(H_i) = N_i$ and there is an edge between two points $(p,i)\not= (q,i)$ in $H_i$ if there exists at least one $e \in E_i$ between two  descendants of $(p,i)$ and $(q,i)$, respectively.  Recall that a high degree point $(p,i)$ has at least $\frac{4c}{\epsilon}$ neighbors in $H_i$.  We proceed in two steps. 

\begin{itemize}
	\item \textbf{Step 1}  For every low degee point $(p,i)$, we add all incident edges of $(p,i)$ in $H_i$ to $E_{sp}$.
	\item \textbf{Step 2} Let $Q$ be the set of high degree points in $N_i$. We add to $E_{sp}$ the set of edges of $\stp(Q)$. We take from each high degree cover point $(p,i)$ exactly $\eps L_i/2$ uncharged descendants and let $X$ be the set of these uncharged points. We charge the cost of  $\stp(Q)$ equally to all points in $X$ and mark them \emph{charged}. This charging is possible by property (2) of $T$. 
\end{itemize}

\subsection{Bounding the stretch}
We will show that the stretch is $(1+(48c+1)\epsilon)$. We can recover stretch $(1+\epsilon')$ by setting $\epsilon' = \frac{\epsilon}{48c+1}$.

Observe by construction that for every edge $e \in H_i$, the stretch of $e$ in $E_{sp}$ is at most $(1+\epsilon)$. Recall that for every edge $(u,v) \in E_i$, there is an edge $(p,q) \in E(H_i)$ such that $u  \in \desc(p,i), v\in \desc(q,i)$. By Claim~\ref{clm:ET}, there is a path between $u$ and $v$ in $E_{sp}$ of length at most $d_{E_{sp}}(p,q) + 8c\epsilon L_i$. By triangle inequality, $||p,q||-2c\epsilon L_i \leq ||u,v||  \leq ||p,q||+2c\epsilon L_i $. Thus, we have:

\begin{equation*}
\begin{split}
\frac{d_{E_{sp}}(u,v)}{||u,v||} &\leq  \frac{d_{E_{sp}}(p,q) + 8c\epsilon L_2}{||p,q|| - 2c\epsilon L_i} = \frac{d_{E_{sp}}(p,q)}{||p,q||} + \frac{(d_{E_{sp}}(p,q) + 4||p,q||)2c\epsilon L_i}{||p,q||(||p,q|| - 2c\epsilon L_i)}\\
&\leq  \frac{d_{E_{sp}}(p,q)}{||p,q||} + \frac{12c\epsilon L_i}{||p,q|| - 2c\epsilon L_i} \qquad \mbox{since }d_{E_{sp}}(p,q) \leq 2||p,q||\\
&\leq   \frac{d_{E_{sp}}(p,q)}{||p,q||} + 48c\epsilon \leq 1 + (48c+1)\epsilon
\end{split}
\end{equation*}
The penultimate inequality follows from the fact that $$||p,q|| - 2c\epsilon L_i ~\geq~ ||u,v|| -4c\epsilon L_i ~\geq~ L_i/2 - 4c\epsilon L_i ~\geq~ L_i/4$$ when $\epsilon \ll \frac{1}{c}$.  


\subsection{Bounding $w(E_{sp})$} Observe that the total number of low degree points in $N_i$ is at most $\frac{|K\cup P|}{\epsilon L_i }$ since each low degree point has at least $\epsilon L_i$ (uncharged) descendants by property (1) of charging tree $T$. By triangle inequality, each edge of $H_i$ has weight at most $L_i + 2c\epsilon L_i \leq 3L_i$ when $\epsilon < \frac{1}{c}$. Thus, the total weight of edges added to $E_{sp}$ in Step 1 is bounded by:
\begin{equation*}
\frac{3L_i \cdot 4c}{\epsilon} |\{p | p \in N_i \mbox{  and $p$ has low degree}\}| = O( \frac{L_i}{\epsilon} \frac{|K\cup P|}{\epsilon L_i }) = O(\frac{1}{\epsilon^2})w(\mst)  
\end{equation*}
by Equation~\ref{eq:mst-vs-points}. Thus, the total weight of  the edges added to $E_{sp}$ in Step 1 over $m$ levels is $O(\epsilon^{-2}\log_{\frac{1}{\epsilon}} \Delta)w(\mst) $. Note here that $m \leq  \lceil \log_{\frac{1}{\epsilon}}\Delta \rceil + 1$.

We now bound the total weight of the edges added to $E_{sp}$ in Step 2 \emph{over $m$ levels}. Since each edge of $H_i$ has weight at most $3L_i$,  each edge of $\stp(Q)$ has weight a most $(1+\epsilon) 3L_i \leq 6L_i$.  By Theorem~\ref{thm:LS19-high-dim},
\begin{equation*}
w(\stp(Q)) = \tilde{O}(\epsilon^{-(d-1)/2}) |Q| 6L_i =  \tilde{O}(\epsilon^{-(d-1)/2 + o(1)}) |Q| L_i
\end{equation*}
Thus, in Step 2, each uncharged point is charged at most:
\begin{equation}
\frac{\tilde{O}(\epsilon^{-(d-1)/2}) |Q| L_i }{|Q| \epsilon L_i/2} =  \tilde{O}(\epsilon^{-(d+1)/2}) 
\end{equation}
This implies the total weight of the edges added to $E_{sp}$ in Step 2 \emph{over all levels} is $\tilde{O}(\epsilon^{-(d+1)/2}) (|P\cup K|) = \tilde{O}(\epsilon^{-(d+1)/2})w(\mst)$. Together with Claim~\ref{clm:ET}, we conclude that:
\begin{equation*}
w(E_{sp})  = \tilde{O}(\epsilon^{-(d+1)/2} + (\delta + \epsilon^{-2} )\log_{\frac{1}{\epsilon}} \Delta)w(\mst)
\end{equation*}
This completes the proof of Lemma~\ref{lm:lightSP-cover-tree}.

\subsection{Constructing a Charging Cover Tree}\label{app:charging-cover}

The main difficulty is to  guarantee property (1) in constructing  a charging cover tree; for property (2) at each level $i$, we simply charge to exactly $\epsilon L_i/2$ uncharged descendants of each high degree cover point. 

A natural idea   is to guarantee that each cover point has $\frac{1}{\epsilon}$ children. Then inductively, if each child of a cover point $(p,i)$ has at least $\epsilon L_{i-1}/2$ uncharged descendants after the charging at level $i-1$, we can hope that $p$ has at least $\frac{1}{\epsilon} \epsilon L_{i-1} = \epsilon L_i$ uncharged descendants. There are two issues with this idea: (a)  if at least one child, say $(q,i-1)$, of $(p,i)$ has high degree in the graph $H_{i-1}$, up to $\epsilon L_{i-1}/2$ uncharged points in $\desc(q,i-1)$ were charged at level $i-1$ by property (2), and thus, $(q,i-1)$ only contributes $\frac{\epsilon L_{i-1}}{2}$ uncharged descendants to $(p,i)$; and (b) there may not be enough cover points at level $i-1$ close to $p$ as these points and their descendants must be within distance $c\epsilon L_{i}$ from $p$. 

In our construction, we resolve both issues by picking a cover point in a way that the total number of uncharged descendants of its children is at least $\epsilon L_i$. We do so by having a more accurate way to track the number of uncharged descendants of a cover point, instead of simply relying on the lower bound $\epsilon L_i$ of uncharged descendants. Specifically, denote by $\D(X)$ the diameter of a point set $X$. We will construct a charging cover tree in a way that the following invariant is maintained \emph{at all levels}.
\begin{quote}
	\textbf{Strong Charging Invariant:} (SCI) 
	
	Each point $(p,i)$ has at least $\max(\epsilon L_i, \D(\desc(p,i))+1)$  uncharged descendants (before the charging happened at level $i$).
\end{quote}

Clearly, SCI implies property (1) of $T$. We begin by constructing the level-$1$ cover. Recall that $\mst$ edges have weight at most $1$ and at least $\frac{1}{2}$, and that $\delta \geq 1$.

\subsubsection*{Level $1$} We construct level-$1$ cover points $N_1$ by greedily breaking $\mst$ edges into subtrees of diameter at least $\delta$ and at most $3\delta + 2\leq 5\delta$. Let $X$ be such a subtree of $\mst$ with diameter $d_X$; $X$ will have at least $\delta$ points since $\mst$ edge has weight at most $1$. We pick any point, say $p \in X$, to be a level-$1$ cover point, and make other points in $X$ become $p$'s children;  $p$ will have at least $\delta$ children (uncharged at level $0$). Recall that $\epsilon L_1 = \delta$. Since each $\mst$ edge has weight at most $1$,  the number of descendants of $(p,1)$ is at least:
\begin{equation*}
d_X + 1 \geq \max(\epsilon L_1, \D(\desc(p,1))+1)
\end{equation*}
Thus, SCI holds for this level.

\subsection*{Level $i+1$} Recall $H_i$ is a graph with vertex set $N_i$.   We construct a cover $N_{i+1}$ in three steps A, B and C.

\begin{itemize}
	\item \textbf{Step A} For each high degree point $p$ (with at least $\frac{4c}{\epsilon}$ unmarked neighbors in $H_i$), we  pick $p$ to $N_{i+1}$  andmake  its unmarked neighbors become its children. We then mark $p$ and all of its neighbors. For each remaining unmarked high degree point $x$ in $H_i$, at least one of its neighbors, say $q$, must be marked before. We make $x$ become a child of $q$'s parent. 
\end{itemize}

The intuition of the construction in Step A is that $(p,i+1)$ picked at this step has at least $\frac{4c}{\epsilon}$ children. Since at least $(\epsilon L_i)/2$  descendants of each child of $(p,i)$ remain uncharged after level $i$, the total number of uncharged descendants of $(p,i+1)$ is $\frac{4c}{\epsilon} \frac{\epsilon L_i}{2} = 2cL_i = 2c\epsilon L_{i+1}>\epsilon L_{i+1}$.  Furthermore, since every high degree point is marked in this step, points in subsequent steps have low degree  and hence no uncharged descendant of these points is charged at level $i$ by property (2) of charging cover trees.

Let $W$ be the set of remaining points in $N_i$. We construct a forest $F$ from $W$ as follows. The vertex set of $F$ is $W$, and there is an edge between two vertices $p,q$ of $F$ if  there is an $\mst$ edge, called the source of the edge, connecting a point in $\desc(p)$ and a point in $\desc(q)$.  We set the weight of each edge in $F$ to be the weight of the source edge.  Note that $F$ is not a geometric graph. We observe that:

\begin{observation}\label{obs:structure-F}
	For every connected component $C \in F$, there must be a point $p\in C$ and a point $q$ marked in Step A such that there is an $\mst$ edge between a descendant of $(p,i)$ and a descendant of $(p,i)$, except when there is no point marked in Step A (and $F$ is a tree in this case).
\end{observation}
\begin{proof}
	The observation follows from the fact that $\mst$ spans $P\cup K$. 
\end{proof}

We define the weight on each vertex $p$ of $F$ to be $w(p) = D(\desc(p,i))$, and the \emph{vertex-weight} of a path $P$, denoted by $\vw(P)$, of $F$ to be the total weight of vertices on the path. We define the \emph{absolute weight} of $P$, denoted by $\aw(P)$, to be the total vertex and edge weight of $P$.  Since each $\mst$ edge has weight at most $1$ and each vertex has weight at least $1$,
\begin{equation}\label{eq:aw-vs-vw}
\aw(P) \leq 2\vw(P)
\end{equation}

The \emph{vertex-diameter} of  a subtree $C$, denoted by $\VD(C)$, of $F$ is defined to be the vertex-weight of the path of maximum vertex-weight in the subtree. The \emph{absolute diameter} of $C$, denoted by $\AD(C)$, is defined similarly but w.r.t absolute weight. 

\begin{itemize}
	\item  \textbf{Step B} For each component $C$ of $F$ of vertex-diameter at least $L_{i}$, we greedily break $C$ into sub-trees of vertex-diameter at least $L_i$ and at most $3L_i$.  For each subtree of $C$, choose an arbitrary point $p$ to be a level-$(i+1)$ cover point and make other points become $p$'s children.
	
	\item  \textbf{Step C} For each component $C$ of $F$ of  vertex-diameter at most $L_i$, by Observation~\ref{obs:structure-F}, there must be at least one $\mst$ connecting a point in $\desc(u,i)$ for some $u \in C$ to a point in $\desc(v,i)$ for some point $v$ marked in  Step 1. We make all points in $C$ become children of $v$'s parent.
	
\end{itemize}
The following claim implies that $N_{i+1}$ is a $(c\epsilon L_{i+1})$-cover.

\begin{claim}\label{clm:cover} For each point $(p,i+1)$, $\D(\desc(p,i+1)) \leq c\epsilon L_{i+1}$ for $c = 20$.
\end{claim}
\begin{proof}
	Note that for each point $(q,i)$, $1\leq \D(\desc(q,i)) \leq 2c\epsilon L_{i}$ since every point in $\desc(q,i)$ is within distance $c\epsilon L_i$ from $q$.  
	
	First, consider the case that $(p,i+1)$ is chosen in Step B. Then $p$ and its children belong to a subtree $C$ of $F$ of vertex-diameter at most $3L_i$. by Equation~\ref{eq:aw-vs-vw},  $\AD(C) \leq 2\cdot 3L_i = 6L_i$. Thus, for a point $p \in N_{i+1}$ selected in Step B, $D(\desc(p,i+1)) \leq \AD(C) \leq 6L_i$.

	We now consider the case  where $(p,i+1)$ is chosen in Step A. (There is no cover point at level $(i+1)$ selected in Step C.)  Recall that each edge of $H_i$ has length at most $L_i + 2c\epsilon L_i$ by the triangle inequality. Observe that after Step 1, for any $(q,i) \in \child(p,i+1)$, the hop distance in $H_i$ between $(p,i)$ and $(q,i)$ is at most $2$, hence $||p, q|| \leq 2(L_i + 2c\epsilon L_i) = 2L_i  + 4c\epsilon L_i$. That implies, after Step A, 
	\begin{equation}\label{eq:dist-StepA}
	||p,x|| \leq ( 2L_i + 4c\epsilon L_i) + c\epsilon L_i = 2L_i + 5c\epsilon L_i  \quad \mbox{for any $x \in \desc(p,i+1)$}
	\end{equation}
	
	In Step C, we add more points belonging to subtrees of $F$ to $\child(p,i+1)$. Let $C$ be any of these subtrees. Since $\VD(C) \leq L_i$, $\AD(C) \leq 2L_i$. By construction, there exists a point $(v,i) \in C$ and a point $(u,i) \in \child(p,i+1)$ such that there is an $\mst$ edge $e$ connecting a point in $\desc(v,i)$ and a point in $\desc(u,i)$. Thus, the augmentation in Step C increases the distance from $p$
	to the furthest point of $\desc(p,i+1)$ by at most $w(e) + \AD(C)~\leq~ 1 + 2L_i ~\leq~ 3L_i$. This implies:
	\begin{equation*}
	\D(\desc(p,i+1)) \leq 2(2L_i + 5c\epsilon L_i + 3L_i) \leq 20L_i
	\end{equation*}
	when $\epsilon < \frac{1}{c}$.
\end{proof}
To complete the proof of Lemma~\ref{lm:lightSP-cover-tree}, it remains to show SCI for level $(i+1)$.

\begin{claim}\label{clm:SCI} Each point $(p,i+1)$ has at least $\max(\epsilon L_{i+1}, \D(\desc(p,i+1)+1)$ uncharged descendants. 
\end{claim}
\begin{proof}
	We first consider the case $(p,i+1)$ is picked in Step B. Let $X$ be the subtree of $F$  that $p$ belongs to. Let $P$ be a path of $X$ of maximum absolute weight. By definition on absolute weight, $\aw(P) \geq \D(\desc(p,i+1))$.  Since $\mst$ has length at most $1$, we have:
	\begin{equation*}
	\sum_{q\in P}{\D(\desc(q,i))} + |E(P)|\geq \aw(P) \geq \D(\desc(p,i+1))
	\end{equation*}
	By SCI for level $i$, we conclude that the number of uncharged descendants of  $(p,i+1)$ is at least:
	\begin{equation*}
	\sum_{q\in P}(\D(\desc(q,i)) +1) \geq (\sum_{q\in P}{\D(\desc(q,i))} + |E(P)|) + 1 \geq \D(\desc(p,i+1))) + 1 
	\end{equation*}	
	To show that $X$ has at least $\epsilon L_{i+1}$ uncharged descendants, we observe by construction that $X$ has a path $Q$ with $\vw(Q) \geq L_i$. By definition of vertex-weight, $\vw(Q) = \sum_{q\in Q}\D(\desc(q,i))$. Thus, the total number of uncharged descendants of all $q \in Q$ by SCI is at least:  
	\begin{equation*}
	\sum_{q \in Q}\left(\D(\desc(q,i)) + 1\right)> \sum_{q \in Q}\D(\desc(q,i)) \geq L_i = \epsilon L_{i+1}
	\end{equation*}
	Thus, $(p,i+1)$ has at least $\max(\epsilon L_{i+1}, \D(\desc(p,i+1))+1)$ uncharged descendants.	
	
	It remains to consider the case $(p,i+1)$ is picked in Step A.  By construction, $(p,i+1)$ has at least $\frac{4c}{\epsilon}$ children, and each has at least $\epsilon L_{i-1}/2$ uncharged descendants by property (2) of a charging cover tree. Note that $\D(\desc(p,i+1)) \leq c\epsilon L_{i+1} = cL_i $ by Claim~\ref{clm:cover}.
	Thus, $(p,i+1)$ has at least:
	\begin{equation*}
	\frac{4c}{\epsilon} \frac{\epsilon L_{i}}{2} = 2cL_{i}\geq \max(\epsilon L_{i+1}, \D(\desc(p,i+1))) + cL_{i} >  \max(\epsilon L_{i+1}, \D(\desc(p,i+1))+1)  
	\end{equation*}
	uncharged descendants as desired.
\end{proof}


\bibliographystyle{plain}
\bibliography{spanner}

\appendix
\end{document}